\documentclass{article}
\usepackage{arxiv}

\usepackage{amsmath,amsfonts,amssymb}
\usepackage{graphicx}
\usepackage[dvipsnames]{xcolor}
\usepackage{hyperref}
\usepackage{natbib}
\usepackage{tikz}
\usepackage{tabularx,ragged2e,booktabs,caption}
\usepackage{subcaption}
\usepackage{microtype}
\usepackage{enumitem}
\usepackage{float}
\usepackage{placeins}
\hypersetup{
    colorlinks=true,
    linkcolor=MidnightBlue,
    citecolor=ForestGreen,
    urlcolor=MidnightBlue
}

\usepackage{amsthm}
\theoremstyle{plain}
\newtheorem{theorem}{Theorem}
\newtheorem{lemma}{Lemma}
\newtheorem{corollary}{Corollary}

\newtheorem{proposition}{Proposition}
\theoremstyle{definition}
\newtheorem{definition}{Definition}
\newtheorem{example}{Example}

\title{Brownian Motion with a Pulse:\\A Biostatistician's Guide to Diffusions, Bridges, Functional PCA, and First-Passage Models}

\author{
  Elvis Han Cui\\[0.25em]
  {\normalfont\small Department of Biostatistics, University of California, Los Angeles, Los Angeles, CA, USA}\\
  {\normalfont\small Kuntuo, an IQVIA company, China}\\
  {\normalfont\small \texttt{elviscuihan@g.ucla.edu}}
}

\date{}

\begin{document}
\maketitle

\begin{abstract}
Brownian motion is a compact mathematical language for continuous-time uncertainty in biostatistics. This tutorial develops the process from construction and path properties to tools that recur in applied biomedical work: the Markov and strong Markov properties, the Karhunen-Loève expansion, functional principal component analysis (Functional PCA), reflection principles, local time, stochastic differential equations (SDEs), Brownian bridges, and empirical-process limits. The applications emphasize longitudinal biomarkers, degradation modelling, first-passage endpoints, dynamic frailty, group-sequential monitoring, calibration diagnostics, recurrent-event processes, electronic health records, and wearable streams. A short cross-domain section uses literary and historical archives to make Brownian-bridge thinking concrete without shifting the paper away from biostatistics, and includes a reproducible chapter-level experiment on \emph{Frankenstein}. The Black-Merton-Scholes model is included as a solved SDE template, not as a finance application in its own right. The aim is to connect rigorous probability with modelling decisions faced by biostatisticians when biological processes evolve between noisy observation times.
\end{abstract}


\section{Introduction}

Brownian motion is a canonical model for continuous-time random fluctuation. Its paths are continuous yet nowhere differentiable, its increments are independent and Gaussian, and its covariance kernel generates both diffusion models and empirical-process limits. These properties make it a useful organizing object for biostatistical problems where the scientific process evolves continuously but is observed intermittently, noisily, and sometimes only after clinical decisions have already been made.

This geometry is familiar across biomedical data analysis. A latent biomarker changes between clinic visits; organ damage accumulates until a threshold is crossed; a trial statistic moves in information time as interim looks accrue; a survival model leaves martingale residuals whose cumulative pattern should look like noise if the fitted hazard is adequate; a wearable device records densely but imperfectly while the physiological state underneath remains hidden. In each setting, Brownian motion is not a claim that biology is literally Gaussian. It is a disciplined baseline for the stochastic part of the model.

This manuscript is written as a tutorial and methods note for biostatisticians. We begin with the formal definition of Brownian motion and the construction of continuous sample paths, then move through the Markov and strong Markov properties, the Karhunen-Loève expansion, Functional PCA, the reflection principle, zero sets, local time, Lévy's characterization, stochastic differential equations, the Black-Merton-Scholes template, and Donsker's theorem \citep{donsker1952justification,billingsley2013convergence}. The proofs stay close to classical probability theory, while the examples translate each tool into longitudinal, survival, clinical-trial, risk-prediction, EHR, wearable, and degradation-modelling language.

The paper makes three contributions. First, it gives a compact route through the Brownian results most useful for applied modelling: construction, stopping, reflection, local time, martingale characterization, SDEs, first-passage distributions, and empirical-process convergence. Second, it links the Karhunen-Loève expansion to Functional PCA for sparse or dense biomarker curves, so that a theorem about covariance operators becomes a practical dimension-reduction tool. Third, Theorem~\ref{thm:biostats-bridge} formalizes a Brownian-bridge diagnostic for ordered biomedical residuals, with a null limit, an ordered-drift alternative, and a consistent drift locator. A brief set of literary and historical examples, including one reproducible chapter-level experiment, is included as a teaching device for the same ordering and bridge ideas.

Figure~\ref{fig:biostats-gallery} gives the practical grammar. In the settings considered here, Brownian motion is a working model for uncertainty in motion rather than an ornamental source of random variation.

\begin{figure}[t]
    \centering
    \includegraphics[width=\textwidth]{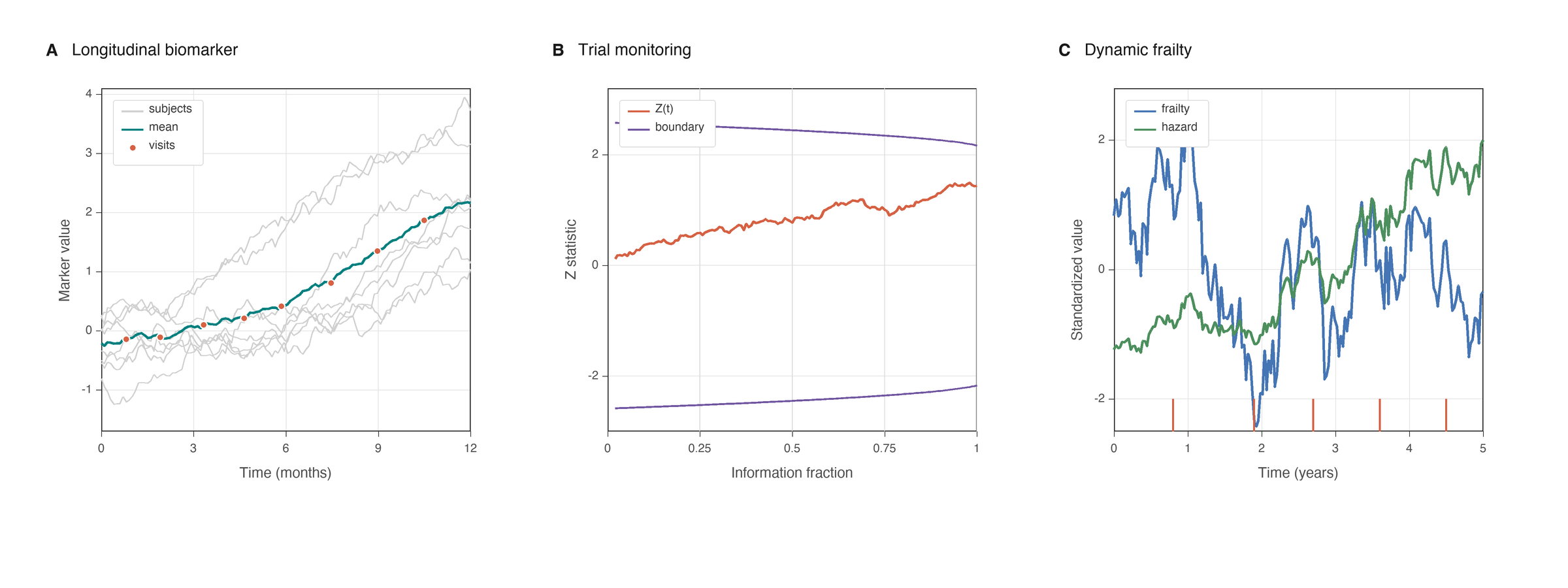}
    \caption{Biostatistical paths.}
    \label{fig:biostats-gallery}
\end{figure}

\section{Definitions}

The first stop is deliberately formal. Before Brownian motion can be used
for biomarkers, residual processes, trial monitoring, or degradation
endpoints, its randomness must be specified independently of any one
application. The definitions below build that path from increments,
covariance, sample space, and filtration. Later, when a longitudinal
trajectory or an ordered residual curve is compared with a Brownian limit,
this formal object becomes the reference shape.

\subsection{Brownian Motion}
\begin{definition}[Standard Brownian Motion \citep{liggett2010continuous}]
Standard Brownian Motion $B(t)$ is a stochastic process with continuous paths that satisfies one of the following properties:
\begin{itemize}
    \item[(a)] \( B(t) \) has stationary independent increments, and \( B(t) \) is \( \mathcal{N}(0,t) \) for \( t \geq 0 \).
    \item[(b)] \( B(t) \) is a Gaussian process with \( \mathbb{E}B(t) = 0 \) and 
    \[
    \text{Cov}(B(s), B(t)) = s \wedge t.
    \]
\end{itemize}
\end{definition}

\textbf{Remarks}: 
\begin{itemize}
    \item Sometimes we write $B_t$ for $B(t)$ and $B_t(\omega)$ or $B(t,\omega)$ for a given sample point $\omega\in\Omega$.
    \item There exists a probability space $(\Omega,\mathcal{H},P)$ on which standard Brownian motion $B$ exists. The proof is postponed to Theorem~\ref{thm:exist}.
\end{itemize}

It is convenient to take $\Omega=C{[0,\infty)}$, the space of all continuous functions $\omega(\cdot)$ on $\mathbb{R}^+$ or $[0,\infty)$. Note that if we simply take $\Omega=\mathbb{R}^{[0,\infty)}$ and its Borel $\sigma$-algebra, then the set
$$C=\{\omega:B(t,\omega)\text{ is continuous in }t.\}$$
is not even an event since it depends on uncountably many points. For $\Omega=C{[0,\infty)}$, the $\sigma$-algebra is taken to be the Borel $\sigma$-algebra, i.e., the smallest one for which the projection mapping $\omega\rightarrow\omega(t)$ is measurable for each $t$ and we denote it as $\mathcal{H}$. We have a family $\mathbb{P}^\bullet=\{P^x\}$ of probability measures indexed by $x\in\mathbb{R}^1$ instead of a single probability measure. The measure $P^x$ is the distribution of $x+B(\cdot)$ where $B$ is a standard Brownian motion. The corresponding expectation is denoted by $\mathbb{E}^x$. The probability triple is denoted by $(\Omega,\mathcal{H},\mathbb{P}^\bullet)$. If in addition, we have a filtration $\mathbb{F}=\{\mathcal{F}_t:t\ge 0\}$, then the stochastic basis is the quadruple $(\Omega,\mathcal{H},\mathbb{F},\mathbb{P}^\bullet)$ \citep{dabrowska2020stochastic}.

\begin{lemma}
Let $\Omega=\mathbb{R}^{[0,\infty)}$ and $\mathcal{H}$ be its Borel $\sigma$-algebra. Then the set
$$C=\{\omega:B(t,\omega)\text{ is continuous in }t.\}$$
is not even an event, i.e., $C\not\in\mathcal{H}$.
\end{lemma}

\begin{proof} 
Define property $(*)$: $\exists J\text{ countable s.t. if } f\in A,g\in \mathbb{R}^{[0,\infty)}$, then $f|_J=g|_J$ implies $g\in A.$ Define the family of sets
$$\mathcal{A}=\left\{A\subseteq\mathbb{R}^{[0,\infty)}: A\text{ satisfies }(*).\right\}$$
Clearly $\mathcal{B}_\alpha=\pi_\alpha^{-1}(\mathcal{B}(\mathbb{R}^\alpha))$, is a subset of $\mathcal{A}$ ($\alpha\text{ finite}$ and $\pi_\alpha$ is the projection from $\mathbb{R}^{[0,\infty)}$ to $\mathbb{R}^\alpha$). But $\mathcal{H}$ is generated from $\mathcal{B}_\alpha$:
$$\mathcal{H}=\sigma\left(\bigcup_{\alpha\text{ finite}}\pi_\alpha^{-1}(\mathcal{B}(\mathbb{R}^\alpha))\right).$$
So if we can show that $\mathcal{A}$ is a $\sigma$-algebra and $C\not\in\mathcal{A}$, then we are done. $C$ does not satisfy $(*)$: if $J$ exists, we pick $x\in J^c$ and set $g(x)=f(x)+1$; then $f|_J=g|_J$ but $g$ is not continuous.
\end{proof}

\section{Basic and Advanced Properties}

Once the process exists, the next question is how it behaves. Brownian
motion remembers its present but not its past, crosses and recrosses
levels, spends measurable time near points, and moves continuously while
refusing to have an ordinary derivative. These are not merely technical
curiosities. They are the path-level facts that make Brownian motion useful
for biomedical modelling: threshold times, near-cutoff burden, cumulative
residual drift, and the uncertainty that grows between sparse observations
all inherit their meaning from these properties.

\subsection{Existence and Construction}
\begin{theorem}[Existence and Construction of Brownian motion]\label{thm:exist}
There exists a probability space $(\Omega,\mathcal{H},\mathbb{P})$ on which standard Brownian motion $B$ exists and has continuous sample path.
\end{theorem}

\textbf{Remarks}: Recall that a Brownian motion is a process with the following properties:
\begin{itemize}
    \item[(i)] $B(t)$ has stationary and independent increments;
    \item[(ii)] $B(t)$ is Gaussian with mean $0$ and variance $t$;
    \item[(iii)] $B(t)$ is continuous in $t$.
\end{itemize}
The existence of a process $B(t)$ with $(i)$ and $(ii)$ is guaranteed by Kolmogorov's extension theorem. Here we show that starting with such $B(t)$, we can construct another process with $(i),(ii)$ and continuous path for $t\in[0,1]$. The extension to $t\in[0,\infty)$ is similar.

\begin{proof}[Proof. \citep{dabrowska2020stochastic}]
For $n\ge 1$, define
$$B^n(t) = B(t)\text{ if }t=\frac{k}{2^n},\ k=0,1,\cdots,2^n$$
and connect points using straight lines, i.e., if $s\in[\frac{k-1}{2^n},\frac{k}{2^n}]$, then
$$B^n(s)=B\left(\frac{k-1}{2^n}\right)+{2^n}\left(B\left(\frac{k}{2^n}\right)-B\left(\frac{k-1}{2^{n}}\right)\right)\left(s-\frac{k-1}{2^n}\right).$$

Define
$$A_n=\left\{ \sup_{t\in[0,1]} \Big| B^n(t)-B^{n-1}(t) \Big| >\frac{1}{n^2}\right\}$$
so that
\begin{align*}
    P\left(A_n\right)&=P\left(\max_{k\le 2^n}\Big| B^n\left(\frac{k}{2^n}\right)-B^{n-1}\left(\frac{k-1}{2^n}\right) \Big| >\frac{1}{n^2}\right)\\
    &\le 2^nP\left(\Big|\mathcal{N}(0,2^{-n})\Big|>\frac{1}{n^2}\right)
    =2^nP\left(\Big|\mathcal{N}(0,1)\Big|>\frac{2^{n/2}}{n^2}\right)\\
    &\le\frac{3n^8}{4^n}\text{ by Markov's inequality.}
\end{align*}

By the ratio test, we have
$$\frac{P(A_{n+1})}{P(A_{n})}\rightarrow\frac{1}{4}<1$$
so 
$$\sum_{n=1}^\infty P(A_n)<\infty.$$
By Borel-Cantelli, the event $\limsup _nA_n$ has probability $0$ and 
$$P\left( \sup_{t\in[0,1]}\Big| B^n(t)-B^{n-1}(t) \Big| \le\frac{1}{n^2}\right)=1$$
for all but finitely many $n$. But $B^n(t)$ is uniformly continuous on $t\in[0,1]$ and forms a Cauchy sequence in $C[0,1]$ with the $\sup$ norm. Hence, their limit is continuous.
\end{proof}

\subsection{The Markov and Strong Markov Property}
Define two $\sigma$-algebras:
\begin{align}
    \mathcal{F}_s^0&=\sigma(B_r:r\le s)\\
    \mathcal{F}_s^+&=\bigcap_{t>s}\mathcal{F}_t^0
\end{align}
Let $f(u)>0$ for all $u>0$, then the random variable
$$\lim\sup_{t\downarrow s}\frac{B_t-B_s}{f(t-s)}\in\mathcal{F}_s^+$$
but is not $\mathcal{F}_s^0$-measurable since it depends on the value of $B_t$. The $\sigma$-algebra $\mathcal{F}_0^+$ can be considered as  ``an infinitesimal peak" at the future \citep{durrett2019probability}.

\begin{theorem}[Markov Property]
If $Y$ is bounded and measurable, then for every $x\in\mathbb{R}^1$ and $s\ge 0$, we have
\begin{align}
    \mathbb{E}^x(Y\circ\theta_s|\mathcal{F}_s^+)&=\mathbb{E}^{B(s)}Y\text{ a.s. }P^x
\end{align}
\end{theorem}

\textbf{Remarks}: The right-hand side is the composition of the function $y\rightarrow\mathbb{E}^yY$ with $B(s)$. Note that the function $y\rightarrow\mathbb{E}^yY$ is measurable for $Y$ bounded.

\begin{proof}
Since $y\rightarrow\mathbb{E}^yY$ is measurable, its composition with $B(s)$ is also measurable. Hence, it suffices to show that the right-hand side satisfies 
$$\mathbb{E}^x(Y\circ\theta_s\mathbb{I}_A)=\mathbb{E}^x((\mathbb{E}^{X(s)}Y)\mathbb{I}_A)$$
for $Y$ bounded and $A\in\mathcal{F}_s^+$.
\end{proof}

\begin{theorem}[Strong Markov Property]
If $Y_s(\omega)$ is bounded and jointly measurable on $[0,\infty)\times\Omega$, and $\tau$ is a stopping time, then for every $x\in\mathbb{R}^1$ and $s\ge 0$, we have
\begin{align}
    \mathbb{E}^x(Y_\tau\circ\theta_\tau|\mathcal{F}_\tau)&=\mathbb{E}^{B(\tau)}Y_\tau\text{ a.s. }P^x\text{ on }\{\tau<\infty\}.
\end{align}
\end{theorem}

\textbf{Remarks}: The expression $Y_\tau\circ\theta_\tau$ means the mapping $\omega\rightarrow Y_{\tau(\omega)}(\theta_{\tau(\omega)}\omega)$. The right-hand side is the function $\phi(y,t)=\mathbb{E}^yY_t$ evaluated at $y=B(\tau)$ and $t=\tau$. We can replace the Brownian motion $B(\tau)$ by any Markov process with the Feller property (a Feller process). The theorem says that $B_{t+\tau}-B_\tau$ is also a Markov process for any fixed $t$ and stopping time $\tau$, and is independent of the stopped $\sigma$-algebra $\mathcal{F}_\tau$.

\subsection{Karhunen-Loève Expansion}

\begin{theorem}[Karhunen-Loève \citep{dabrowska2020stochastic}] 
If $\{X_t:t\in[0,1]\}$ is a process with zero mean and finite variance, then it admits a decomposition
\begin{align}
    X_t &= \sum_{n=1}^\infty Y_n\phi_n(t),
\end{align}
where $Y_n$'s are pairwise uncorrelated random variables and $\phi_n$'s form an orthonormal basis in $L^2([0,1])$ determined by $K(s,t) = \text{Cov}(X_s,X_t)$, the covariance of the process.
\end{theorem}

\textbf{Remarks}: If $K(s,t)$ is a positive definite kernel and the associated eigenfunctions form a complete orthonormal sequence of $L^2([0,1])$, then we can take $Y_n = \langle X, \phi_n \rangle$ where $\phi_n$'s are eigenfunctions of $K$. In this case, $Y_n$'s are uncorrelated mean zero variables with variance $\lambda_n$, the eigenvalues of $K(s,t)$.

\begin{example}[Brownian Motion on the unit interval]
Standard Brownian Motion on $[0,1]$ has covariance $K(s,t) = s \wedge t$ and can be decomposed as
\[
K(s,t) = \sum_{j=1}^\infty \lambda_j \phi_j(s)\phi_j(t)
\]
where ($j\ge 1$)
\begin{align*}
    &\lambda_j = \frac{4}{(2j-1)^2\pi^2},\ \phi_j(t) = \sqrt{2}\sin\left(\left(j-\frac{1}{2}\right)\pi t\right).
\end{align*}
Moreover, Brownian motion has the same distribution as
\[
W_t \sim \sum_{j=1}^\infty\sqrt{\lambda_j}V_j\phi_j(t),
\]
where $V_j, j\ge 0$ are i.i.d. $\mathcal{N}(0,1)$ variables. The derivation of $\phi_j$ boils down to solving differential equations and its distribution relies on the convergence of Riemann sums.
\end{example}

\begin{example}[Functional PCA for longitudinal biomarkers]
For a biostatistician, the Karhunen-Loève expansion is the probability engine behind functional principal component analysis (FPCA) \citep{ramsay2005functional,yao2005functional}. Suppose \(X_i(t)\) is a centered subject-specific biomarker trajectory on \([0,1]\), such as standardized glucose, creatinine, viral load, or activity intensity. If
\[
K(s,t)=\operatorname{Cov}\{X_i(s),X_i(t)\},
\]
then
\[
X_i(t)=\sum_{j=1}^{\infty}\xi_{ij}\phi_j(t),
\qquad
\operatorname{Var}(\xi_{ij})=\lambda_j,
\]
where \(\phi_j\)'s are the eigenfunctions of \(K\). With noisy and irregular clinic visits, the observed data are often closer to
\[
Y_{ij}=\mu(t_{ij})+\sum_{k=1}^{K}\xi_{ik}\phi_k(t_{ij})+\epsilon_{ij},
\]
where the FPCA scores \(\xi_{ik}\) summarize the patient's latent trajectory. In practice, only a few scores may be needed to describe clinically meaningful shapes: baseline elevation, early rebound, late deterioration, or cyclic instability. These scores can then enter a regression, clustering analysis, joint longitudinal-survival model, or risk prediction tool. If the working covariance is Brownian, \(K(s,t)=s\wedge t\), the eigenfunctions are sine waves with decreasing eigenvalues. If the working covariance is a Brownian bridge, \(K(s,t)=s\wedge t-st\), the endpoint is anchored, which is useful for residual trajectories constrained to return to zero after centering.
\end{example}

\section{Path Properties, Diagnostics, and Applications}

\subsection{A Biostatistical Field Guide}

Brownian motion enters biostatistics most cleanly when the scientific object is continuous in time but observed through a coarse, noisy schedule. Clinic visits are discrete, assays are noisy, and adverse events arrive at inconvenient moments; the latent biological process does not care. Brownian models give us a disciplined way to let the hidden process move between observations.

\begin{table}[H]
\centering
\small
\caption{Recurring Brownian motifs in biostatistical modeling.}
\label{tab:biostats-field-guide}
\begin{tabularx}{\textwidth}{>{\RaggedRight\arraybackslash}p{0.22\textwidth}>{\RaggedRight\arraybackslash}p{0.23\textwidth}>{\RaggedRight\arraybackslash}p{0.27\textwidth}>{\RaggedRight\arraybackslash}X}
\toprule
Model & Typical data & Brownian ingredient & Scientific use \\
\midrule
Latent biomarker diffusion & Irregular longitudinal measurements, assay error, missing visits & 
\(dY_i^\ast(t)=\{m_i(t)-\kappa_iY_i^\ast(t)\}dt+\sigma_i dB_i(t)\), \(Y_{ij}=Y_i^\ast(t_{ij})+\epsilon_{ij}\) &
Separates biological movement from measurement noise; useful for disease progression and treatment response. \\
Functional PCA for biomarker curves & Sparse longitudinal visits, dense wearables, smoothed lab trajectories &
\(Y_{ij}=\mu(t_{ij})+\sum_{k=1}^{K}\xi_{ik}\phi_k(t_{ij})+\epsilon_{ij}\), with \(\phi_k\)'s from the covariance operator &
Compresses each subject's curve into interpretable scores for phenotyping, prediction, and downstream survival or treatment-response models. \\
Dynamic frailty survival model & Event times with time-varying covariates and unmeasured risk &
\(\lambda_i(t)=\lambda_0(t)\exp\{Z_i^\top\gamma+U_i(t)\}\), \(dU_i(t)=-\rho U_i(t)dt+\eta dB_i(t)\) &
Lets unobserved inflammation, adherence, or vulnerability drift over follow-up rather than remain fixed. \\
Sequential clinical trial monitoring & Interim test statistics indexed by information fraction &
\(Z(t)=\theta t+B(t)\), with stopping when \(Z(t)\) crosses a boundary \(b(t)\) &
Turns efficacy or safety monitoring into a first-passage problem; this is the geometry behind many group-sequential designs \citep{lan1983discrete}. \\
Empirical-process diagnostics & Distributional checks, calibration curves, nonparametric estimators &
\(\alpha_n(t)=\sqrt n\{F_n(t)-F(t)\}\Rightarrow B(F(t))-F(t)B(1)\) &
Explains Kolmogorov-Smirnov, Cramér-von Mises, and calibration-band behavior through the Brownian bridge. \\
Recurrent-event intensity & Hospitalizations, infections, exacerbations, device alerts &
\(N_i(t)-\int_0^t\lambda_i(s)ds\) as a martingale with Brownian limits after scaling &
Links counting-process residuals to Gaussian fluctuation limits for model checking and surveillance. \\
Wearable or EHR state-space model & Dense sensor traces, irregular lab panels, noisy vital signs &
\(dX_i(t)=\mu_i(t)dt+\sigma_i(t)dB_i(t)\), \(Y_{ij}=X_i(t_{ij})+\epsilon_{ij}\) &
Lets the analyst smooth noisy observation streams without pretending the body moves only at visit times. \\
Degradation modeling & Biomarker deterioration, device wear, organ-function decline, assay drift &
\(D_i(t)=D_i(0)+\nu_i t+\sigma_iB_i(t)\), with failure time \(\tau_i=\inf\{t:D_i(t)\ge c\}\) &
Turns progressive damage into a first-passage problem, linking longitudinal decline to time-to-threshold endpoints \citep{kahle2016degradation}. \\
SDE and Black-Merton-Scholes template & Continuous-time stochastic models, Itô calculus, diffusion-based likelihoods &
\(dS_t=\mu S_tdt+\sigma S_tdB_t\), giving \(S_t=S_0\exp\{(\mu-\sigma^2/2)t+\sigma B_t\}\) &
Provides the canonical SDE solved by Itô's formula; useful as a transferable template before building biomedical diffusions. \\
\bottomrule
\end{tabularx}
\end{table}

This table is intentionally pragmatic. Longitudinal and joint-model pieces let a biomarker trajectory move as a noisy diffusion while an event process reads that trajectory through a hazard model \citep{diggle2002analysis,henderson2000joint,rizopoulos2012joint}. Counting-process notation makes martingales and predictable variation the bookkeeping system for censored survival and recurrent-event data \citep{andersen1982cox,andersen1993statistical,fleming1991counting,aalen2008survival}. Trial monitoring turns the Brownian first-passage picture behind the reflection principle into a stopping rule \citep{lan1983discrete,jennison2000group}. Empirical-process diagnostics turn goodness-of-fit, calibration, and distributional checking into Brownian-bridge geometry \citep{donsker1952justification,billingsley2013convergence,vandervaart1996weak,kosorok2008introduction}. EHR and wearable studies use state-space diffusions less as molecular truth than as principled interpolation: uncertainty grows during long gaps and tightens when dense observations arrive \citep{kalman1960new,durbin2012time}. The unifying point is that these are not separate tricks; they are different uses of the same stochastic structure.

\subsection{Blumenthal's 0-1 Law}

\begin{theorem}[Blumenthal's 0-1 Law \citep{liggett2010continuous}]
\begin{itemize}
    \item[(a)] If $Y$ is a bounded random variable, then for every $x$
    \begin{align}
        \mathbb{E}^x(Y|\mathcal{F}_s^+) = \mathbb{E}^x(Y|\mathcal{F}_s^0)\text{ a.s. }P^x.
    \end{align}
    \item[(b)] If $A \in \mathcal{F}_0^+$, then
    \begin{align}
        P^x(A) = 0\text{ or }1.
    \end{align}
\end{itemize}
\end{theorem}

\begin{proof}
For part (a), it is enough to consider $Y$ of the form 
\[
Y = \prod_{m=1}^n f_m(B(t_m))
\]
where $f_m$ are bounded and measurable. In this case, we write
\[
Y = Z(X\circ\theta_s),
\]
where $Z = \prod_{m:t_m\le s} f_m(B(t_m))$ and $X = \prod_{m:t_m>s} f_m(B(t_m-s))$. Then $Z\in\mathcal{F}_s^0 \subset\mathcal{F}_s^+$ and
\[
\mathbb{E}^x(X\circ\theta_s|\mathcal{F}_s^+) = \mathbb{E}^{B(s)}X \text{ a.s.}
\]
by the Markov property. Taking conditional expectation $\mathbb{E}^x(\cdot|\mathcal{F}_s^0)$ gives $\mathbb{E}^x(X\circ\theta_s|\mathcal{F}_s^0)=\mathbb{E}^{B(s)}X$.

For part (b), set $Y = \mathbb{I}_A$, and then
\[
\mathbb{I}_A = \mathbb{E}^x(Y|\mathcal{F}_0^+) = P^x(A|\mathcal{F}_0^+) = P^x(A|\mathcal{F}_0^0),
\]
where the last equality is from part (a). But given the initial information $\mathcal{F}_0^0$ and the initial state $x$, the set $A$ either contains $x$ or not. If it contains, then $\mathbb{I}_A = 1$ a.s.; otherwise, $\mathbb{I}_A = 0$ a.s.
\end{proof}

\begin{corollary} 
$P^0(\tau=0) = 1$ in each of the following two cases:
\begin{itemize}
    \item[(a)] $\tau = \inf\{t > 0 : B(t) > 0\}$.
    \item[(b)] $\tau = \inf\{t > 0 : B(t) = 0\}$.
\end{itemize}
\end{corollary}

\begin{proof}
In both cases, the event $\{\tau=0\}\in\mathcal{F}_0^+$ but not in $\mathcal{F}_0^0$. Hence, by Blumenthal's 0-1 law, it is enough to show the event has positive probability. For part (a), given any positive $t$, $P^0(\tau\le t) \ge P^0(B(t) > 0) = \frac{1}{2}$. Then we take $t\downarrow 0$. For part (b), we just replace $B(t)>0$ with $B(t)<0$ in part (a) and apply the path continuity of $B(\cdot)$. In other words, we can pick a sequence of positive numbers that goes to $0$: $l_1 > t_1 > r_1 > l_2 > t_2 > r_2 > \cdots$ such that $B(l_i)>0$, $B(t_i)=0$ and $B(r_i)<0$ for $i=1,2,\cdots$.
\end{proof}

    We have an analogue of Kolmogorov's 0-1 law for Brownian motion.
    \begin{theorem}[Kolmogorov's 0-1 Law for Brownian Motion] Define the tail $\sigma$-algebra
    $$\mathcal{T}=\bigcap_{t>0}\mathcal{G}_t$$
    where $\mathcal{G}_t=\sigma(B_s:s\ge t)$.
        Then for any $A\in\mathcal{T}$, we have $P^x(A)=0$ or $1$, independent of $x$.
    \end{theorem}

\subsection{Zero Set}

Consider the zero set (for a given $\omega\in\Omega$)
\begin{align}
    Z(\omega) &= \left\{t\in\mathbb{R}_+ : B_t(\omega)=0\right\}
\end{align}

\begin{theorem}[Zero Set]
With probability 1, the set $Z$ is perfect, and hence closed and uncountable. The Lebesgue measure of $Z$ is $0$ a.s.
\end{theorem}

\textbf{Remark}: A set is perfect if it is closed with no isolated points, i.e., every point is a limit point. By the Markov property, we conclude that the level set $L(a,\omega)=\{t\in\mathbb{R}_+ : B_t(\omega)=a\}$ has the same properties as the zero set $Z(\omega)$ does.

Figure~\ref{fig:zero-local-time-sim} gives a simulation-based picture of
this theorem. On a finite grid we can only mark sign-change intervals, not
the full zero set. Even so, the rug plot already suggests the central
paradox: the path returns to zero repeatedly, while the amount of ordinary
Lebesgue time spent exactly at zero is still negligible.

\begin{proof}
The $0$ Lebesgue measure is easy to show:
\[
\mathbb{E}m(Z) = \mathbb{E}\int_{\mathbb{R}_+}\mathbb{I}_Z(t)dt = \int_{\mathbb{R}_+}\mathbb{E}\mathbb{I}_Z(t)dt = \int_{\mathbb{R}_+}\mathbb{P}(B(t)=0)dt = 0
\]
For a given $\omega$, $Z = B^{-1}(\{0\})$, i.e., the inverse image of $\{0\}$ under the continuous mapping $B(\cdot,\omega)$. Hence, $Z$ is closed and its complement is a countable union of open intervals. We write
\[
Z^c = \bigcup_{n\in\mathbb{N}}(l_n,r_n)
\]
and $r_n$ are stopping times but $l_n$ are not.

Next, let $t\in Z$ be any zero of the path, then $t$ is either a limit point of $Z$ from the left or not.

\begin{itemize}
    \item Case I: If $t$ is a limit point from the left, then $t$ is not isolated.
    \item Case II: If $t$ is not a limit point from the left, then there exists $\epsilon>0$ so that there is no zero within the interval $(t-\epsilon,t)$. Define the random variable
    \[
    T_a(\omega) = \inf\{t\ge a : B(\omega, t) = 0\},
    \]
    where $a$ is any positive rational number. It is a stopping time:
    \[
    \{T_a \le s\} = 
    \begin{cases}
        \emptyset & \text{if } s < a\\
        \left(\left\{B(t)>0,t\in[a,s]\right\}\cap\left\{B(t)<0,t\in[a,s]\right\}\right)^c & \text{if else}
    \end{cases}
    \]
\end{itemize}

By Strong Markov Property, given the stopping field $\mathcal{F}_{T_a}$, $B\circ \theta_{T_a}$ is another standard BM where $\theta$ is the shift operator. Since for a standard BM, given any $\epsilon>0$, there are infinitely many $0$'s within the interval $[T_a, T_a+\epsilon)$. Hence, there must be a rational $a\in(t-\epsilon,t)$ such that $t = T_a$. In this case, $t$ is a limit point of $Z$ from the right.
\end{proof}

\subsection{Non-differentiability}

We start with a proposition that describes the peculiarities of the Brownian motion path, and then jump into pointwise and global non-differentiability of $B$.

\begin{proposition}
Almost surely the Brownian motion $B$ is not monotone on any interval $[s,t]$.
\end{proposition}

\begin{proof}
It is enough to show that the following set has probability $0$:
\begin{align*}
    A &= \bigcup_{s,t\in \mathbb{Q^+}}\left\{ \omega\in\Omega:B(\omega)\text{ is monotone on [s,t]} \right\}\\
    &= \bigcup_{s,t\in \mathbb{Q^+}}A_{st}
\end{align*}
By \(\sigma\)-additivity and stationarity of \(B\), it suffices to prove the representative case
\[
A_{01}=\{B\text{ is monotone on }[0,1]\}
\]
has probability \(0\).

Let 
\[
B_n = \bigcap_{i=1}^n\left\{ B\left(\frac{i-1}{n}\right)-B\left(\frac{i}{n}\right)\ge 0 \right\}
\]
so that $A_{01} = \cap_{n=1}^\infty B_n$. But
\[
P^0(B_n) = 2^{-n} \rightarrow 0 \text{ as }n\rightarrow\infty.
\]
\end{proof}

\begin{lemma}[Pointwise Nonsmoothness] 
For a fixed $t\ge 0$, 
\begin{align}
    \mathbb{P}(B(\cdot,\omega)\text{ is not differentiable at }t) = 1.       
\end{align}
Further, for $\mathbb{P}-a.s.$ $\omega\in\Omega$,
\begin{align}
    m(A) = 0,\ A = \{t\ge0:B(\cdot,\omega)\text{ is differentiable at }t\}.
\end{align}
\end{lemma}

\begin{proof}
For $t = 0$, we have $\lim\sup_{t\downarrow0}\frac{B(t)}{\sqrt{t}} = +\infty$ a.s. If $B$ is differentiable at $B$, then $\exists K,\ \epsilon>0$, and we have (by the mean value theorem)
\[
|B(s)-B(0)| \le K |s-0|, s\in[0,\epsilon).
\]
A contradiction. Next, for any $t>0$, we note that $B(t)-B(t_0)$ has the same distribution as $B(t-t_0)$.  

For the second part, fix $\omega$ and define the random set
\[
C = \left\{t\ge0 : \lim_{n\rightarrow\infty}\frac{B(t+\frac{1}{n},\omega)-B(t,\omega)}{1/n}\text{ exists at $t$.}\right\}
\]
By Fubini's theorem and joint measurability of $B(t,\omega)$
\begin{align}
    \mathbb{E}\left(m(A)\right) &\le \mathbb{E}m(C) = \mathbb{E}\left(\int_0^\infty\mathbb{I}_C(t)dt\right) = \int_0^\infty\int_\Omega\mathbb{I}_C(t)d\mathbb{P}dt = 0
\end{align}
where the last equality follows from the first part.
\end{proof}

\begin{theorem}[Nondifferentiability of $B$: Paley-Wiener-Zygmund]
The Brownian motion $B(\cdot,\omega)$ is nowhere differentiable a.s.
\end{theorem}

\textbf{Remarks}: By the Markov property, $B(\cdot+s)-B(s)$ has the same distribution as $B(\cdot)$, so it is sufficient to show that $\mathbb{P}(B\text{ is differentiable on }(0,1]) = 0.$ Let $D_s$ be the set such that the path $B(\cdot,\omega)$ is differentiable at $s$. By the previous lemma, we know $\mathbb{P}(D_s) = 0$ for any $s$. Here we want to show that 
\[
D = \bigcup_{s\in(0,1]}D_s = \{\omega : B(s,\omega)\text{ is differentiable at some }s\in(0,1]\}
\]
has $\mathbb{P}$-measure $0$, but this is an uncountable union. {The idea here is to bound this quantity by discretizing the sample path of $B$ and each piece has a small oscillation. Then $D$ can be written as a countable union of sets. Within each discretized interval, the sample path cannot be Lipschitz continuous and hence, cannot be differentiable.}

\begin{proof}
Recall 
\[
D = \bigcup_{s\in(0,1]}D_s
\]
where $D_s$ is the set that the path $B(\cdot,\omega)$ is differentiable at $s$. Define 
\begin{align}
    \Gamma &= \bigcup_{m= 1}^\infty\liminf_{n\rightarrow \infty}\underbrace{\bigcup_{k=1}^{n-2}\bigcap_{j=k}^{k+2}\left\{\Big|B\left(\frac{j}{n}\right)-B\left(\frac{j-1}{n}\right)\Big|\le\frac{3m}{n}\right\}}_{D_{mn}}\\
    &= \bigcup_{m= 1}^\infty\bigcup_{l=1}^\infty\bigcap_{n=l}^\infty D_{mn}\nonumber
\end{align}

\textcolor{black}{If $D \subseteq \Gamma$}, then we have
\[
\mathbb{P}(D) \le \mathbb{P}(\Gamma) \le \sum_{m=1}^\infty\mathbb{P}\left(\liminf_{n\rightarrow\infty}D_{mn}\right) = \sum_{m=1}^\infty\mathbb{P}\left(\bigcup_{l=1}^\infty\bigcap_{n\ge l}^\infty D_{mn}\right)
\]
and for any fixed $m$. The proof is completed if we can show each term on the right-hand side is 0. By stationarity and independence of increments,
\begin{align*}  
\mathbb{P}\left(\bigcup_{l=1}^\infty\bigcap_{n=l}^\infty D_{mn}\right) &\le \liminf_{n\rightarrow\infty} \mathbb{P}\left(D_{mn}\right)\\
&\le \liminf_{n\rightarrow\infty} \sum_{k=1}^{n-2}\mathbb{P}\left( \Big| B\left(\frac{1}{n}\right) \Big|\le\frac{3m}{n} \right)^3
\end{align*}
From the previous slide, we have
\begin{align*}
    \mathbb{P}\left(\bigcup_{l=1}^\infty\bigcap_{n=l}^\infty D_{mn}\right) &\le \liminf_{n\rightarrow\infty} n\left[\mathbb{P}\left( \Big| B\left(1\right) \Big|\le\frac{3m}{\sqrt{n}} \right)\right]^3\\
    &\le \liminf_{n\rightarrow\infty} n\left(\int_{-\frac{3m}{\sqrt{n}}}^{\frac{3m}{\sqrt{n}}}\frac{1}{\sqrt{2\pi}}e^{-\frac{x^2}{2}}dx\right)^3\\
    &\le\liminf_{n\rightarrow\infty}\left(\frac{6m}{\sqrt{2\pi }}\right)^3\frac{1}{\sqrt{n}}\\
    &=0.
\end{align*}
We only have $D \subseteq \Gamma$ left to show.

\begin{lemma}
We have $D \subseteq \Gamma$ where $D = \bigcup_{s\in(0,1]}D_s$ and 
\[
\Gamma = \bigcup_{m= 1}^\infty\liminf_{n\rightarrow \infty} D_{mn}
\]
where 
\[
D_{mn} = \bigcup_{k=1}^{n-2}\bigcap_{j=k}^{k+2}\left\{\Big|B\left(\frac{j}{n}\right)-B\left(\frac{j-1}{n}\right)\Big|\le\frac{3m}{n}\right\}.
\]
\end{lemma}

\begin{proof}
If $B$ is differentiable at some $s\in(0,1]$, then it is Lipschitz continuous at $s$. Let
\[
A_{mn} = \left\{\exists s\in(0,1],\ \Big|B(t)-B(s)\Big|\le m|t-s|\ \forall t\text{ s.t. }|t-s|\le\frac{2}{n} \right\}
\]
so that clearly
$D\subseteq\bigcup_{m}\bigcup_n A_{mn}.$ We next show that for any $\omega\in D$, then $\omega\in\Gamma$.

If $\omega\in D$, then $\exists M,N$ so that $\omega\in A_{MN}$. 

Pick $k$ s.t. $\frac{k}{N}\le s\le\frac{k+1}{N}$, so $\omega\in\bigcap_{j=k}^{k+2}\left\{\Big|B(\frac{j}{N})-B(\frac{j-1}{N})\Big|\le\frac{3M}{N}\right\}$.

In other words, we can pick $N_0$ large and 
\[
\omega\in \bigcap_{N\ge N_0}\bigcup_{k=1}^{N-2}\bigcap_{j=k}^{k+2}\left\{\Big|B\left(\frac{j}{N}\right)-B\left(\frac{j-1}{N}\right)\Big|\le\frac{3M}{N}\right\} = \bigcap_{N\ge N_0}D_{MN}.
\]
We have shown that if $\omega\in D$, then $\omega \in \Gamma = \bigcup_{m}\bigcup_{l}\bigcap_{n\ge l}D_{mn}$.
\end{proof}
\end{proof}

We conclude this subsection with a precise statement of the exact modulus of continuity of the Brownian motion.

\begin{proposition}[Lévy's Modulus of Continuity] 
Define
\begin{align}
    g(\delta) &= \sqrt{2\delta\log\left(\frac{1}{\delta}\right)},\ \delta\in(0,1],
\end{align}
we have
\begin{align}
    P^0\left(\limsup_{\delta\downarrow0}\frac{1}{g(\delta)}\sup_{\substack{0\le s,t\le1\\ |t-s|<\delta}}\Big|B(t)-B(s)\Big|\right) = 1.
\end{align}
\end{proposition}

For a proof, see page 114 in \cite{karatzas2014brownian} or theorem 7.13 in \cite{ccinlar2011probability}.

\subsection{Reflection Principle}

This is another application of the strong Markov property.

\begin{theorem}[Reflection Principle]
If $a>0$ and $b<a$, then for any $t>0$,
\begin{align}
    P^0\left(M(t)>a,B(t)<b\right) = P^0\left(B(t)>2a-b\right),
\end{align}
where 
\[
M(t) = \max_{0\le s\le t}B(s)
\]
is the running maximum of $B$ up to time $t$.
\end{theorem}

\begin{proof}
Define 
\begin{align*}
    &\tau = \inf\{s\ge 0 : B(s)=a\},\\
    &Y_s(\omega) = \mathbb{I}(s<t,B(t-s)>2a-b)-\mathbb{I}(s<t,B(t-s)<b).
\end{align*}
For $s<t$,
\[
\mathbb{E}^x Y_s = P^x(B(t-s)>2a-b)-P^x(B(t-s)<b).
\]
By symmetry, on the set $\{\tau<\infty\}$, we have
\[
\mathbb{E}^{B(\tau)} Y_\tau = \mathbb{E}^a Y_\tau = \mathbb{E}^a Y_s |_{s=\tau} = 0.
\]

By the strong Markov property, on the set $\{\tau<\infty\}$
\[
\mathbb{E}^0\left(Y_\tau\circ\theta_\tau\ \mathbb{I}(\tau<\infty)\right) = \mathbb{E}^0\left(\underbrace{\mathbb{E}^a\left(Y_\tau\circ\theta_\tau\ |\mathcal{F}_\tau\right)}_{\mathbb{E}^{B(\tau)}Y_\tau}\mathbb{I}(\tau<\infty)\right) = 0.
\]

On the other hand, on the set $\{\tau<\infty\}$

\[
 Y_{\tau(\omega)}(\theta_{\tau(\omega)}\omega) = \begin{cases}
    +1 & \text{if }\tau(\omega)<t,\ B(t)>2a-b;\\
    -1 & \text{if }\tau(\omega)<t,\ B(t)<b;\\
    0 & \text{if else}.
\end{cases}
\]

Hence,
\[
\mathbb{E}^0\left(Y_\tau\circ\theta_\tau\ \mathbb{I}(\tau<\infty)\right) = P^0(\tau<t,B(t)>2a-b)-P^0(\tau<t,B(t)<b).
\]
The event $\{\tau<t\}$ iff $\{M(t)>a\}$ by the path continuity of $B$; and we have $\{B(t)>2a-b\}\subseteq\{\tau<t\}$. In conclusion,
\[
P^0(B(t)>2a-b) = P^0(M(t)>a,B(t)<b).
\]
\end{proof}

\textbf{Remarks}: A more useful formula is derived as $P^0(|B(t)|>a) = P^0(M(t)>a)$, see next corollary.

\begin{corollary}
Based on the reflection principle, we have
\begin{itemize}
    \item Relative to $P^0$, $M(t)$, $|B(t)|$ and $M(t)$ have the same distribution.
    \item Let $\tau_a = \inf\{t>0 : B(t)=a\}$ be the first hitting time of $a$, then
    \begin{align}
        P^0(\tau_a\le T) &= 2\left(1-\Phi\left(\frac{a}{\sqrt{T}}\right)\right).
    \end{align}
\end{itemize}
\end{corollary}

\begin{proof}
\begin{itemize}
    \item Let $a>0$ and $b\uparrow a$,
    \begin{align*}
        P^0(M(t)>a) &= P^0(M(t)>a,B(t)< a)+P^0(M(t)>a,B(t)\ge a)\\
        &= 2P^0(B(t)>a).
    \end{align*}
    \item Since $\{\tau_a\le T\}=\{M(T)\ge a\}$, 
    \begin{align*}
       P^0(\tau_a\le T) &= P^0(M(T)\ge a).
    \end{align*}
\end{itemize}
\end{proof}

\subsection{Local Time}

The \textbf{occupation time} (or the occupation measure) of a Borel set $A\subseteq \mathbb{R}^1$ by Brownian motion up to time $t$ is
\begin{align*}
    K_t(\omega, A) = \int_0^t\mathbb{I}_A(B(s,\omega))ds = m(\{s\in[0,t] : B(s,\omega)\in A\}),
\end{align*}
which is a random measure of total mass $t$. It is absolutely continuous w.r.t. the Lebesgue measure and its Radon-Nikodym derivative is called the \textbf{local time} \citep{liggett2010continuous}. The Radon-Nikodym derivative is jointly continuous in space and time a.s. and can be conceptually thought/defined as
\begin{align}\label{eq:local_time}
    L_t(a,\omega) = \lim_{\epsilon\downarrow 0}\frac{K_t(\omega, (a-\epsilon,a+\epsilon))}{2\epsilon}.
\end{align}
and the proof of existence is non-trivial.

In other words, the local time $L_t(a,\omega)$ serves as a density w.r.t. $m(\cdot)$, the Lebesgue measure for occupation time \citep{morters2010brownian, karatzas2014brownian}, i.e.,
\[
K_t(\omega, A) = \int_0^t\mathbb{I}_A(B(s,\omega))ds = \int_A dx L_t(x,\omega),\ t\in\mathbb{R}_+,A\in\mathcal{B}(\mathbb{R}).
\]
\textbf{Remarks}: As mentioned on page 203 in \cite{karatzas2014brownian}, there is no universal agreement in the literature whether we should divide the denominator by $2\epsilon$ or $4\epsilon$ in Equation~\ref{eq:local_time}. Here I divide it by $2\epsilon$ (\cite{ccinlar2011probability} did it also), but \cite{karatzas2014brownian,ikeda1984introduction} used $4\epsilon$.

\begin{figure}[H]
    \centering
    \includegraphics[width=\textwidth]{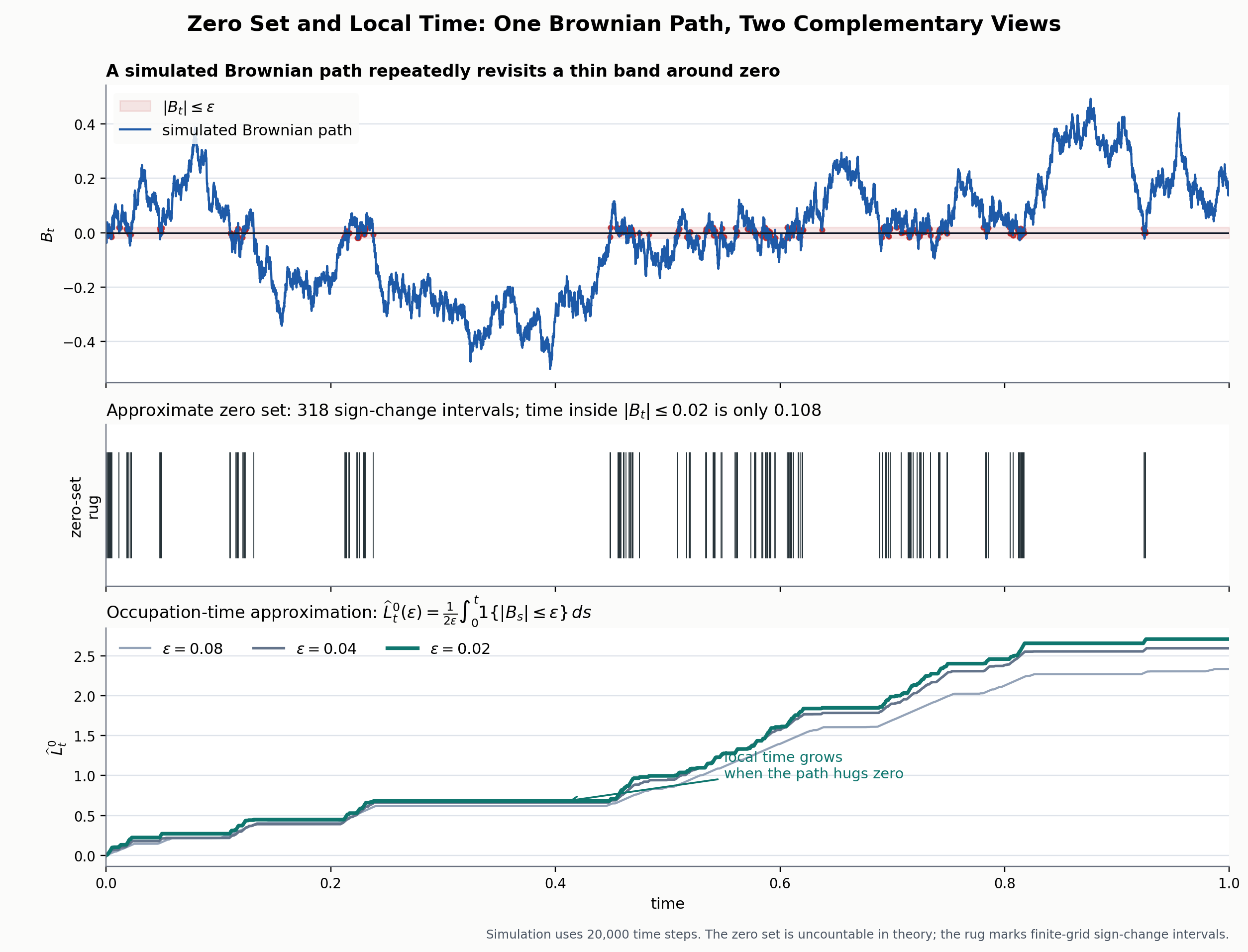}
    \caption{Simulation view of the Brownian zero set and local time at zero. The top panel shows one Brownian path with a thin \(\epsilon\)-tube around zero. The middle panel marks finite-grid sign-change intervals as a rug plot, a computable shadow of the zero set. The bottom panel shows occupation-time approximations \(\widehat L_t^0(\epsilon)=(2\epsilon)^{-1}\int_0^t\mathbb{I}\{|B_s|\le\epsilon\}\,ds\) for three bandwidths. Local time grows during repeated near-zero visits even though the exact zero set has Lebesgue measure zero.}
    \label{fig:zero-local-time-sim}
\end{figure}

\begin{theorem}[Tanaka's Formula for Brownian Local Time] 
For every $a\in\mathbb{R}^1$ there is an increasing continuous stochastic process $L_a(t)$ satisfying
\begin{align}
    |B(t)-a|-|a| = \int_0^t\text{sgn}\left( B(s)-a \right)dB(s) + L_a(t).
\end{align}
The measure $dL_a$ puts all its mass on $\{t : B(t) = a\}$ and $L_a(t)$ is jointly continuous in $(a,t)$.
\end{theorem}

\subsection{Lévy's Martingale Characterization}

\begin{definition}[Predictable Variation \citep{dabrowska2020stochastic}]
If $M(t)$ is a locally square integrable martingale then $M^2(t)$ is a local sub-martingale. By Doob-Meyer's decomposition, there exists a predictable process $\langle M \rangle(t)$ s.t. $M^2(t)-\langle M \rangle(t)$ is a local martingale. The process $\langle M \rangle(t)$ is referred to as the \textbf{predictable variation} of $M(t)$ and can be obtained as the in probability limit of sums
\[
\sum_{i=1}^n\mathbb{E}\left(M^2(t_i)-M^2(t_{i-1})|\mathcal{F}_{t_i-1}\right) = \sum_{i=1}^n\text{Var}\left( M(t_i)-M(t_{i-1})|\mathcal{F}_{t_i-1} \right),
\]
where $0=t_0<t_1<\cdots<t_k=t$ is a partition with mesh $\max_i|t_i-t_{i-1}|\rightarrow 0$.
\end{definition}

\begin{theorem}[Lévy's Characterization of $B$ \citep{ccinlar2011probability}]
If $M(t)$ is a continuous square integrable martingale w.r.t. $\mathbb{F}$ with $M(0)=0$, and $\langle M\rangle(t) = t$, then $M(t)$ is a standard Brownian motion w.r.t. $\mathbb{F}$.   
\end{theorem}

\textbf{Remarks}: By Doob-Meyer, the statement $\langle M\rangle(t) = t$ is equivalent to say that $M^2(t)-t$ is again a martingale w.r.t. $\mathbb{F}$.

\subsection{SDEs and the Black-Merton-Scholes Template}

A stochastic differential equation (SDE) turns Brownian motion from a limiting object into a modeling engine. The basic form is
\[
dX_t=b(t,X_t)dt+\sigma(t,X_t)dB_t,
\]
where \(b\) is the drift and \(\sigma\) is the diffusion coefficient. In biostatistics, \(X_t\) might represent a latent biomarker, a disease severity index, a time-varying frailty term, or a hidden physiological state observed through noisy clinic visits. The notation is compact, but the meaning is concrete: over a short interval \(dt\), the process moves by a systematic component plus a random shock whose variance is proportional to \(dt\).

The canonical solved example is geometric Brownian motion,
\[
dS_t=\mu S_tdt+\sigma S_tdB_t,\qquad S_0>0.
\]
By Itô's formula applied to \(\log S_t\),
\[
d\log S_t=\left(\mu-\frac{\sigma^2}{2}\right)dt+\sigma dB_t,
\]
and hence
\[
S_t=S_0\exp\left\{\left(\mu-\frac{\sigma^2}{2}\right)t+\sigma B_t\right\}.
\]
This is the diffusion backbone of the Black-Merton-Scholes model \citep{black1973pricing,merton1973theory}. Its purpose here is not to turn the manuscript into finance. It is included because every applied statistician who learns SDEs should see one model where Brownian motion, Itô correction, likelihood intuition, and closed-form consequences all appear in one place.

For biomedical data, the same grammar is usually modified rather than copied. A positive biomarker might use a log-diffusion, a bounded physiological score might use a transformed diffusion, and a frailty process might use mean reversion:
\[
dU_i(t)=-\rho U_i(t)dt+\eta dB_i(t).
\]
The Black-Merton-Scholes example is therefore a template: first understand how Brownian noise enters a multiplicative SDE, then replace the scientific object, drift, boundary, and observation model with ones appropriate for health data.

\subsection{Donsker's Theorem}

The narrative now turns from single Brownian paths to data. Donsker's
theorem is the bridge between probability theory and biostatistical
diagnostics: it says that many centered cumulative sums, after the right
scaling, forget their original units and approach a Brownian limit. Once an
endpoint correction is imposed by centering, the Brownian bridge appears.
This is why calibration curves, cumulative score processes, sequential
trial statistics, and survival residual processes can often be interpreted
with the same geometric reference.

In non-parametric statistics, one of the most important distributions is the empirical distribution
\[
F_n(t) = \frac{1}{n}\sum_{i=1}^n\mathbb{I}(X_i\le t),
\]
where $X_i$'s are i.i.d. random variables with the same distribution $F$. By SLLN, $F_n(t)$ converges to $F(t)$ a.s. and the convergence is strengthened to uniform by Glivenko-Cantelli. By the classical CLT, $\sqrt{n}\left(F_n(t)-F(t)\right)$ converges weakly to $\mathcal{N}(0,F(t)(1-F(t)))$. 

Now if we regard $F_n(\cdot)$ as a stochastic process indexed by $t\in\mathbb{R}$, what can we say about the limiting distribution of 
$\{\sqrt{n}(F_n(t)-F(t)) : t\in\mathbb{R}\}$? Moreover, let $\mathcal{H}$ be a class of measurable functions, then what can we say about the distribution of
\[
\mathbb{P}_n h = \frac{1}{n}\sum_{i=1}^n h(X_i),\ h\in\mathcal{H}\ ?
\]
This is the subject of the empirical process theory. Here we focus on the special case $h_t(X_i) = \mathbb{I}(X_i\le t)$.

\begin{theorem}[Donsker's Theorem for Empirical Distribution \citep{donsker1952justification,billingsley2013convergence}]
Let $X_i$'s be i.i.d. variables with a common distribution $F$. Define the process 
\[
W_{n}^F(t) = \sqrt{n}\left( F_n(t) - F(t) \right),\ t\in\mathbb{R}.
\]
Then we have
\[
W_n^F(\cdot) \Rightarrow W(F(\cdot))\text{ in }D(\mathbb{R}),
\]
where $W(x) = B(x) - xB(1),\ x\in[0,1]$ is the Brownian Bridge (BB) process.
\end{theorem}

\textbf{Remarks}: $D(\mathbb{R})$ is the space of right-continuous functions with left limits and is equipped with the $\sup$-norm. The weak convergence of a stochastic process $X_n$ to $X$ means for every bounded continuous $f$, we have $\mathbb{E}f(X_n)\rightarrow\mathbb{E}f(X)$. Modern empirical-process theory extends this bridge logic from threshold indicators to rich function classes \citep{vandervaart1996weak,kosorok2008introduction}. 

\begin{theorem}[Donsker's Theorem for Partial Sum \citep{donsker1952justification}]
Let $X_i$'s be i.i.d. with mean $0$ and variance $1$. Define the partial sum process 
\[
Z_n(t) = \frac{S_{[nt]}}{\sqrt{n}},\ t\in[0,1],
\]
where $S_{[nt]} = \sum_{i=1}^{[nt]}X_i$ is the partial sum. Then
\[
Z_n \Rightarrow B
\]
in $D[0,1]$, where $B$ is the standard Brownian motion.
\end{theorem}

\begin{theorem}[Biostatistical Brownian-Bridge Diagnostic]\label{thm:biostats-bridge}
Let $X_1,\ldots,X_n$ be ordered residual summaries from a biostatistical workflow, such as calibration residuals ordered by predicted risk, martingale residual summaries ordered by event time, site-level safety residuals ordered by calendar time, longitudinal residual increments ordered by visit time, or wearable residual features ordered by monitoring time. Under the stable fitted-model null, suppose
\[
X_k=\mu+\sigma\epsilon_k,\qquad k=1,\ldots,n,
\]
where $\sigma>0$, $\epsilon_k$ are i.i.d., $\mathbb{E}\epsilon_k=0$, and $\operatorname{Var}(\epsilon_k)=1$. Let $\widehat{\sigma}_n^2$ be the sample variance and define
\[
A_n(t)=\frac{1}{\widehat{\sigma}_n\sqrt n}
\left\{\sum_{k=1}^{\lfloor nt\rfloor}X_k
-\frac{\lfloor nt\rfloor}{n}\sum_{k=1}^{n}X_k\right\},\qquad t\in[0,1].
\]
Then
\[
A_n(\cdot)\Rightarrow W^0(\cdot)=B(\cdot)-(\cdot)B(1)
\]
in $D[0,1]$. In particular,
\[
\sup_{0\le t\le1}|A_n(t)|\Rightarrow \sup_{0\le t\le1}|W^0(t)|.
\]
Under a single ordered drift,
\[
X_{n,k}=\mu+\delta\mathbb{I}\{k/n>\tau_0\}+\sigma\epsilon_k,
\qquad 0<\tau_0<1,\quad \delta\ne0,
\]
define the unstandardized bridge
\[
L_n(t)=\frac{1}{n}
\left\{\sum_{k=1}^{\lfloor nt\rfloor}X_{n,k}
-\frac{\lfloor nt\rfloor}{n}\sum_{k=1}^{n}X_{n,k}\right\}.
\]
Then
\[
\sup_{0\le t\le1}|L_n(t)-H_{\delta,\tau_0}(t)|\xrightarrow{p}0,
\]
where
\[
H_{\delta,\tau_0}(t)=
\begin{cases}
-\delta t(1-\tau_0), & 0\le t\le\tau_0,\\
-\delta \tau_0(1-t), & \tau_0<t\le1.
\end{cases}
\]
Consequently, any measurable maximizer
\[
\widehat{\tau}_n\in\operatorname*{arg\,max}_{0\le t\le1}|L_n(t)|
\]
satisfies $\widehat{\tau}_n\xrightarrow{p}\tau_0$.
\end{theorem}

\begin{proof}
For the null statement, write the centered partial-sum process as
\[
S_n(t)=\frac{1}{\sqrt n}\sum_{k=1}^{\lfloor nt\rfloor}\sigma\epsilon_k.
\]
Donsker's theorem gives $S_n(\cdot)\Rightarrow \sigma B(\cdot)$. The mapping $x(t)\mapsto x(t)-tx(1)$ is continuous at continuous limits, and replacing $t$ by $\lfloor nt\rfloor/n$ changes the process by $o_p(1)$ uniformly. Since $\widehat{\sigma}_n\to_p\sigma$, Slutsky's theorem gives $A_n\Rightarrow B(t)-tB(1)$.

For the ordered-drift statement, split $L_n$ into its deterministic mean part and its centered error part. The deterministic part converges uniformly to $H_{\delta,\tau_0}$ by direct summation, with only $O(1/n)$ error from the floor functions. The stochastic part is
\[
\frac{\sigma}{n}\left\{\sum_{k=1}^{\lfloor nt\rfloor}\epsilon_k
-\frac{\lfloor nt\rfloor}{n}\sum_{k=1}^{n}\epsilon_k\right\},
\]
whose supremum over $t\in[0,1]$ is $o_p(1)$ by the uniform law of large numbers for partial sums. Hence $L_n$ converges uniformly in probability to $H_{\delta,\tau_0}$. The absolute tent function $|H_{\delta,\tau_0}|$ has a unique maximum at $\tau_0$. For every $\eta>0$, continuity and uniqueness give a positive gap between $|H_{\delta,\tau_0}(\tau_0)|$ and $\sup_{|t-\tau_0|\ge\eta}|H_{\delta,\tau_0}(t)|$; uniform convergence transfers that gap to $|L_n|$ with probability tending to one. Thus every maximizer of $|L_n|$ lies inside $(\tau_0-\eta,\tau_0+\eta)$ with probability tending to one.
\end{proof}

\textbf{Remarks}: The i.i.d. assumption is the cleanest version for exposition. In real biomedical applications, residuals may be clustered by patient, site, family, clinician, device, batch, or calendar period; the same bridge diagnostic remains valid when the ordered residual process satisfies an appropriate functional central limit theorem and the variance estimator is consistent \citep{billingsley2013convergence,vandervaart1996weak,kosorok2008introduction}. In practice, this theorem should be read as a transparent diagnostic null, not as a claim that clinical data are literally independent.

\subsection{Finite-Sample Calibration and Practical Use}

The asymptotic Brownian bridge is a clean reference, but biomedical datasets
rarely behave like the limiting experiment. Site-level safety series can be
short, EHR streams can be autocorrelated, wearable measurements can be
bursty, and repeated observations within a patient can make naive variance
scaling too optimistic. For a submission-ready analysis, the bridge
diagnostic should therefore be paired with a calibration strategy that
matches the sampling design.

\begin{table}[H]
\centering
\small
\caption{Practical calibration choices for Brownian-bridge diagnostics in biostatistics.}
\label{tab:bridge-calibration}
\begin{tabularx}{\textwidth}{>{\RaggedRight\arraybackslash}p{0.25\textwidth}>{\RaggedRight\arraybackslash}p{0.33\textwidth}>{\RaggedRight\arraybackslash}X}
\toprule
Data setting & Recommended calibration & Interpretation \\
\midrule
Independent validation cohort &
Use the Brownian-bridge limit as a first reference; report Kolmogorov-type bands and sensitivity to variance estimation. &
Appropriate for calibration residuals or simple ordered score processes when dependence is weak. \\
Patients nested in sites or clinics &
Use cluster-robust variance or site-level resampling before forming the bridge. &
Large excursions then reflect systematic drift beyond patient-level noise rather than site composition alone. \\
Repeated measures within patient &
Order patient-level summaries, or use a block/bootstrap procedure that preserves within-subject dependence. &
Prevents dense follow-up in a few patients from dominating the bridge path. \\
Wearable or EHR time series &
Use block resampling, state-space posterior simulation, or model-based simulation under the fitted null. &
Separates physiological excursions from autocorrelation, device artifacts, and missingness patterns. \\
Group-sequential trial statistics &
Index by information fraction and calibrate against the planned alpha-spending or monitoring boundary. &
Makes boundary crossings interpretable as design-controlled first-passage events. \\
\bottomrule
\end{tabularx}
\end{table}

The diagnostic workflow is deliberately simple: choose the ordering,
define residual features, estimate the variance under the fitted null,
plot the bridge path, calibrate the maximum excursion, and then inspect the
estimated drift location. The mathematics supplies a disciplined null
shape; scientific interpretation still depends on the model, measurement
process, and clinical context.

\subsection{Clinical and Public-Health Data Vignettes}

The Brownian bridge is not only a device for asymptotic statistics; it is also a compact diagnostic language for clinical and public-health data. Once a biostatistician chooses an ordering, a centering rule, and a variance scale, cumulative deviations can ask a sharp question: is this fluctuation compatible with stable stochastic noise, or is there a clinical drift hiding in the record?

\begin{table}[H]
\centering
\small
\caption{Biostatistical data streams that pair naturally with Brownian models.}
\label{tab:clinical-vignettes}
\begin{tabularx}{\textwidth}{>{\RaggedRight\arraybackslash}p{0.24\textwidth}>{\RaggedRight\arraybackslash}p{0.31\textwidth}>{\RaggedRight\arraybackslash}X}
\toprule
Setting & Data structure & Brownian question \\
\midrule
Longitudinal disease progression &
Irregular biomarker panels, medication changes, dropout, visit windows &
If \(dX_i(t)=\mu_i(t)dt+\sigma_i dB_i(t)\), when does the latent process first cross a clinically meaningful threshold, and how much crossing-time uncertainty is just sparse follow-up? \\
Functional PCA biomarkers &
Sparse clinic measurements, dense wearable curves, smoothed laboratory trajectories &
Which eigenfunctions and FPCA scores summarize baseline elevation, rebound, deterioration, or cyclic instability, and how should those scores enter prediction or survival models? \\
Clinical trial monitoring &
Interim \(Z\)-statistics, information fractions, efficacy and safety boundaries &
If \(Z(t)=\theta t+B(t)\), which observed boundary crossing is strong evidence of treatment benefit, and which is the expected consequence of repeated looks? \\
Degradation modelling &
Organ-function decline, assay drift, device wear, biomarker damage accumulation &
If \(D_i(t)=D_i(0)+\nu_i t+\sigma_iB_i(t)\), how long until degradation crosses a clinically meaningful failure boundary? \\
SDE and Black-Merton-Scholes bridge &
Continuous-time stochastic models, positive latent states, Itô-calculus training examples &
What does a solved SDE teach us before we build biomedical diffusions with patient-specific drift and observation noise? \\
Survival and recurrent events &
Censoring, risk sets, counting processes, cumulative hazards &
After predictable variation scaling, do martingale residuals behave like Brownian noise, or do they expose time-varying effects, frailty, or model misspecification? \\
Risk prediction and calibration &
Predicted risks, binary outcomes, subgroup ordering, validation cohorts &
Does the cumulative calibration curve stay inside a Brownian-bridge band, or is the model systematically over- or under-confident in a clinically important region? \\
EHR and wearable surveillance &
Dense sensor traces, sparse labs, missingness, batch effects, alert streams &
Is an excursion a real physiological regime change, a device artifact, or ordinary diffusion-scale variability between measurements? \\
\bottomrule
\end{tabularx}
\end{table}

These vignettes are deliberately close to practice. In each case, Brownian motion supplies either a latent continuous-time model, a Gaussian limit for a cumulative residual, or a first-passage geometry for thresholds and stopping rules. The mathematics stays classical, but the target is applied: better uncertainty statements for clinical decisions.

\subsection{Cross-Domain Examples: Literary and Historical Archives}

A biostatistical tutorial can still benefit from examples outside
biomedicine, provided the role of the examples is clear. Literary and
historical archives are useful here because they turn the Brownian bridge
into a visibly ordered diagnostic: choose a sequence, extract a feature,
center and scale it, and ask whether the cumulative deviation looks like
ordinary bridge fluctuation or a structured departure. This is the same
logic a biostatistician uses for ordered residuals, but the units are
chapters, trials, newspaper issues, voyages, inscriptions, or historical
periods rather than patients and visits.
Examples include Project Gutenberg and SPGC novels, the Old Bailey
Proceedings, Chronicling America, SlaveVoyages, and larger world-literature
or world-history resources such as Seshat, CDLI, the Chinese Text Project,
OpenITI, Perseus, ELTeC, and HathiTrust
\citep{projectgutenberg,gerlach2020standardized,oldbaileyonline,chroniclingamerica,slavevoyages,seshat2026,cdli2026,ctext2026,openiti2026,perseus2026,eltec2021,hathitrust2026}.

\begin{table}[H]
\centering
\small
\caption{Literary and historical examples that illustrate Brownian-bridge thinking.}
\label{tab:literary-history-examples}
\begin{tabularx}{\textwidth}{>{\RaggedRight\arraybackslash}p{0.25\textwidth}>{\RaggedRight\arraybackslash}p{0.32\textwidth}>{\RaggedRight\arraybackslash}X}
\toprule
Archive or corpus & Ordered feature path & Brownian-bridge question \\
\midrule
Project Gutenberg or SPGC novels &
Chapters ordered by narrative sequence; features include affect terms, dialogue ratio, scientific vocabulary, sentiment, or lexical diversity. &
Does a novel's cumulative stylistic or thematic path remain compatible with ordinary chapter-to-chapter fluctuation, or does it bend around a narrative rupture? \\
Old Bailey Proceedings &
Trials ordered by hearing date; features include offence type, verdict residuals, punishment severity, gendered language, or witness structure. &
Is an apparent change in legal practice gradual, or does a bridge excursion concentrate around a historical reform, policing shift, or recording change? \\
Chronicling America newspapers &
Issues ordered by publication date; features include topic attention, named entities, public-health vocabulary, war terms, or advertisement density. &
Does cumulative attention to an event exceed ordinary issue-level variation, and where does the attention path begin to depart from baseline? \\
SlaveVoyages and mortality ledgers &
Voyages or records ordered by departure year, route, port, or vessel class; features include mortality residuals, crowding, duration, or destination. &
Can a bridge path identify periods where risk departs from a fitted historical baseline while keeping uncertainty and missingness visible? \\
World-literature and world-history corpora &
Texts, artifacts, polity-period observations, or volumes ordered by chronology, genre, dynasty, author date, or publication year. &
Do cumulative residuals suggest a smooth diffusion of form and vocabulary, or a concentrated shift that merits philological or historical interpretation? \\
\bottomrule
\end{tabularx}
\end{table}

These examples are not intended to make cultural data the object of the
paper. Their purpose is pedagogical. They show that the same Brownian-bridge
geometry used for calibration residuals or survival diagnostics can also
teach analysts to think carefully about ordering, centering, variance
scaling, dependence, and the difference between random fluctuation and a
scientifically meaningful departure.

\subsection{Illustrative Experiment: A Literary Bridge Diagnostic}

To make the cross-domain analogy concrete, I include one small reproducible
experiment. The point is not to turn the manuscript into a digital-humanities
paper. The point is to show, in a setting with an intuitive ordered unit, how
the same bridge statistic used for biomedical residuals behaves after
centering, scaling, and finite-sample calibration.

The empirical unit is the 24 numbered chapters of Mary Shelley's
\emph{Frankenstein}, after excluding the four frame letters and removing the
Project Gutenberg wrapper \citep{gutenbergFrankenstein84}. For chapter \(k\),
let \(h_k\) be hits from a 65-term gothic-affect lexicon and \(c_k\) be hits
from a 16-term creation/science lexicon. With \(w_k\) chapter words, define
\[
    x_k = 1000(h_k-c_k)/w_k,\qquad
    z_k = \frac{x_k-\bar{x}}{s_x},\qquad
    C_j = \frac{1}{\sqrt{24}}\sum_{k=1}^{j}z_k.
\]
Because the \(z_k\)'s are centered, \(C_{24}=0\), so this is a discrete
Brownian-bridge-style path. To avoid relying only on the asymptotic
Kolmogorov reference, I also simulated 20,000 stable length-24 feature
sequences and a single-shift alternative with a 1.25 standard-deviation mean
increase after chapter 6.

\begin{table}[H]
\centering
\small
\caption{Experimental results for the literary Brownian-bridge diagnostic.}
\label{tab:literary-experiment-results}
\begin{tabularx}{\textwidth}{>{\RaggedRight\arraybackslash}p{0.38\textwidth}>{\RaggedRight\arraybackslash}X}
\toprule
Quantity & Result \\
\midrule
Simulation design & \(n=24\), 20,000 Monte Carlo replicates, seed 20260621 \\
Stable-null 95\% cutoff for \(\max_j |A_j|\) & 1.19 \\
Stable-null 99\% cutoff for \(\max_j |A_j|\) & 1.42 \\
Single-shift alternative & Mean increases by 1.25 standard deviations after chapter 6 \\
Power against this shift using the simulated 95\% cutoff & 0.47 \\
Text and units & Project Gutenberg ebook 84; 24 numbered chapters; 69,700 tokens \\
Lexicon totals & 1,175 gothic-affect hits from 65 terms; 310 creation/science hits from 16 terms \\
Early versus later contrast & Chapters 1--6 average \(-2.10\); chapters 7--24 average \(14.91\) gothic-minus-science terms per 1,000 words \\
Observed bridge statistic & \(\max_j |C_j|=1.47\), attained at chapter 6 \\
Finite-sample tail probability & Simulated stable-null tail probability at 1.47 is 0.006 \\
\bottomrule
\end{tabularx}
\end{table}

\begin{figure}[H]
    \centering
    \includegraphics[width=\textwidth]{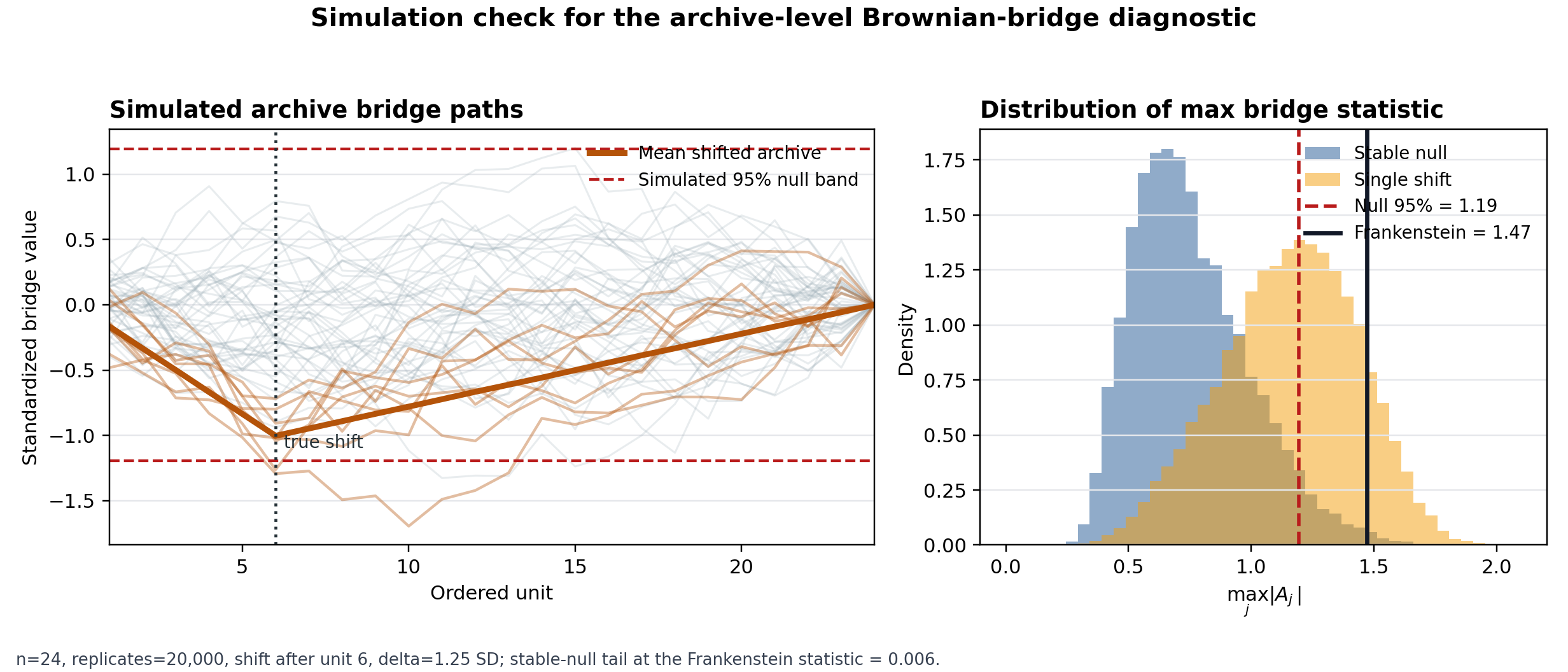}
    \caption{Monte Carlo calibration for the length-24 Brownian-bridge statistic. The stable null provides finite-sample cutoffs, while the single-shift alternative shows the tent-shaped bridge pattern induced by an ordered change.}
    \label{fig:archive-bridge-simulation}
\end{figure}

\begin{figure}[H]
    \centering
    \includegraphics[width=\textwidth]{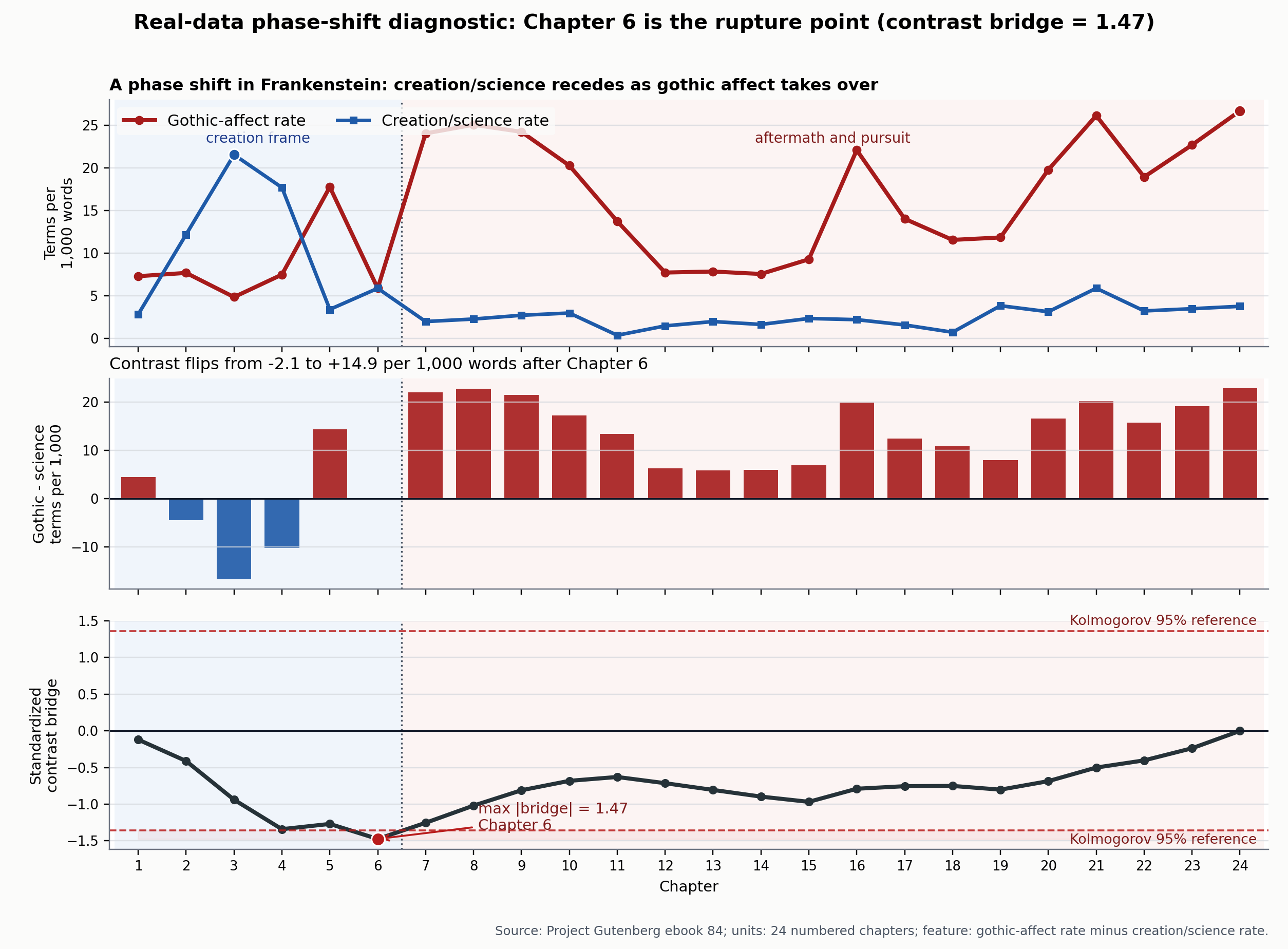}
    \caption{Chapter-level creation-to-catastrophe bridge diagnostic for \emph{Frankenstein}. The top panel compares gothic-affect and creation/science rates, the middle panel shows their signed contrast, and the bottom panel shows the standardized cumulative contrast bridge.}
    \label{fig:frankenstein-gothic-bridge}
\end{figure}

The generic gothic-affect rate is descriptive, but the contrast is more
diagnostic. Creation/science vocabulary dominates the first six numbered
chapters, especially chapter 3; after chapter 6, gothic-affect vocabulary
dominates almost every chapter and peaks in the final pursuit. The bridge
statistic reaches 1.47 at chapter 6, above both the asymptotic 95\% reference
band of 1.36 and the simulated length-24 99\% cutoff of 1.42. In biostatistical
terms, this is an ordered residual experiment with a visible drift location:
the example teaches the bridge geometry without changing the paper's
biomedical target.

\subsection{A Biostatistician's Example Bank}

Examples often illuminate abstract concepts, and Brownian motion is no exception. The examples below keep the same probability backbone but translate it into biostatistical work: Functional PCA, thresholds, degradation, SDEs, survival, trials, calibration, surveillance, and noisy observation streams.

\begin{example}[Functional PCA for sparse biomarker curves]
Suppose \(Y_{ij}\) is a noisy measurement of subject \(i\)'s latent biomarker trajectory at time \(t_{ij}\). A functional PCA model writes
\[
Y_{ij}=\mu(t_{ij})+\sum_{k=1}^{K}\xi_{ik}\phi_k(t_{ij})+\epsilon_{ij},
\qquad
\operatorname{Var}(\xi_{ik})=\lambda_k.
\]
The Karhunen-Loève expansion supplies the population target: the
\(\phi_k\)'s are eigenfunctions of the covariance operator of the latent
process. In a Brownian working model, early eigenfunctions describe broad
trajectory shapes; in a Brownian-bridge working model, the endpoint
constraint is useful for centered residual curves. The scores \(\xi_{ik}\)
can then become covariates in risk prediction, clustering, treatment-effect
heterogeneity analysis, or joint longitudinal-survival models
\citep{ramsay2005functional,yao2005functional,rizopoulos2012joint}.
\end{example}

\begin{example}[Time to a clinical biomarker threshold]
Suppose a latent biomarker follows
\[
X(t)=x_0+\mu t+\sigma B(t),\qquad a=c-x_0>0,
\]
and define the first time it reaches a clinical threshold \(c\) by
\[
\tau_c=\inf\{t>0:X(t)\ge c\}.
\]
For \(\mu=0\), the reflection principle gives
\[
P(\tau_c\le T)=2\left\{1-\Phi\left(\frac{a}{\sigma\sqrt{T}}\right)\right\}.
\]
With nonzero drift, the inverse-Gaussian density is
\[
f_{\tau_c}(t)=\frac{a}{\sigma\sqrt{2\pi t^3}}
\exp\left(-\frac{(a-\mu t)^2}{2\sigma^2t}\right).
\]
For a biostatistician, this is a model for time to renal decline, viral rebound, tumor progression, or loss of immune response when the true biology is only observed at visits.
\end{example}

\begin{example}[Degradation modelling for disease progression]
Let \(D_i(t)\) be accumulated biological damage, device wear, or assay degradation for subject or unit \(i\). A Brownian degradation model is
\[
D_i(t)=D_i(0)+\nu_i t+\sigma_iB_i(t),
\qquad
\tau_i=\inf\{t:D_i(t)\ge c\}.
\]
The random time \(\tau_i\) is a first-passage endpoint: renal function reaches a dangerous decline, a biomarker exceeds a toxicity threshold, or a sensor drifts beyond acceptable calibration. The model is simple enough to yield interpretable hitting-time calculations but flexible enough to connect longitudinal degradation to survival-type endpoints \citep{kahle2016degradation}.
\end{example}

\begin{example}[Black-Merton-Scholes as the solved SDE everyone should see]
Under the Black-Merton-Scholes model,
\[
dS_t=\mu S_tdt+\sigma S_tdB_t.
\]
Itô's formula gives
\[
S_t=S_0\exp\left\{\left(\mu-\frac{\sigma^2}{2}\right)t+\sigma B_t\right\}.
\]
Under risk-neutral pricing, \(\mu\) is replaced by the risk-free rate \(r\), and a European call with strike \(K\) and maturity \(T\) has
\[
C=S_0\Phi(d_1)-Ke^{-rT}\Phi(d_2),
\]
where
\[
d_1=\frac{\log(S_0/K)+(r+\sigma^2/2)T}{\sigma\sqrt T},
\qquad
d_2=d_1-\sigma\sqrt T.
\]
This finance example earns its place because it teaches the mechanics of SDE modeling: multiplicative Brownian noise, Itô correction, transformation, and a closed-form consequence \citep{black1973pricing,merton1973theory}. A biostatistical SDE usually changes the state variable and observation model, but the calculus is the same.
\end{example}

\begin{example}[Group-sequential monitoring as first passage]
Let \(Z(t)\) be a standardized interim statistic indexed by information fraction \(t\in[0,1]\). A common Brownian approximation is
\[
Z(t)=\theta t+B(t),
\]
where \(\theta=0\) under the null and \(\theta\neq0\) under a local alternative. A stopping rule of the form
\[
\tau_b=\inf\{t:Z(t)\ge b(t)\}
\]
is a first-passage problem. Under the null, Brownian boundary approximations explain why repeated looks inflate type I error unless the boundary is designed carefully \citep{lan1983discrete,jennison2000group}.
\end{example}

\begin{example}[Calibration bands for risk prediction]
Let \(R_i\) be a predicted risk and \(Y_i\in\{0,1\}\) an outcome. After sorting patients by predicted risk, a cumulative calibration process can be written schematically as
\[
C_n(u)=\frac{1}{\sqrt n}\sum_{i:R_i\le u}\{Y_i-\widehat p_i\}.
\]
If the model is calibrated and regularity conditions hold, \(C_n(u)\), after variance stabilization and endpoint correction, behaves like a Brownian bridge. Large bridge excursions identify risk regions where the model is clinically misleading. This is the same geometry as the Kolmogorov-Smirnov statistic, but the object is a prediction model a hospital might actually deploy \citep{vandervaart1996weak,kosorok2008introduction}.
\end{example}

\begin{example}[Joint longitudinal-survival modeling]
In a joint model, the latent biomarker trajectory may be written as
\[
dY_i^\ast(t)=\{m_i(t)-\kappa_iY_i^\ast(t)\}dt+\sigma_i dB_i(t),
\qquad
Y_{ij}=Y_i^\ast(t_{ij})+\epsilon_{ij},
\]
while the event hazard reads that latent state:
\[
\lambda_i(t)=\lambda_0(t)\exp\{Z_i^\top\gamma+\alpha Y_i^\ast(t)\}.
\]
The Brownian term is the part of the biology that keeps moving between visits. It lets the survival model borrow information from the longitudinal process without pretending that the last measured lab value is the current value \citep{henderson2000joint,rizopoulos2012joint}.
\end{example}

\begin{example}[Dynamic frailty for survival and recurrent events]
A fixed frailty term says that unmeasured vulnerability is born once and never changes. For many clinical processes, that is too stiff. A dynamic frailty model can use
\[
dU_i(t)=-\rho U_i(t)dt+\eta dB_i(t),
\qquad
\lambda_i(t)=\lambda_0(t)\exp\{X_i(t)^\top\beta+U_i(t)\}.
\]
The Ornstein-Uhlenbeck form keeps frailty from wandering forever while still allowing inflammation, adherence, exposure, or care access to drift over time. Counting-process martingales then provide residual checks for whether the model has absorbed the right amount of temporal heterogeneity \citep{andersen1982cox,andersen1993statistical,fleming1991counting,aalen2008survival}.
\end{example}

\begin{example}[Local time as burden near a clinical cutoff]
Let \(X(t)\) be a centered continuous marker and let \(c\) be a clinically important cutoff. The local time
\[
L_t(c)=\lim_{\epsilon\downarrow0}
\frac{1}{2\epsilon}\int_0^t\mathbb{I}\{|X(s)-c|<\epsilon\}ds
\]
measures how intensely the trajectory occupies the threshold neighborhood. Two patients may have the same number of above-threshold readings, but one repeatedly grazes the cutoff while another makes one decisive excursion. Local time separates near-threshold burden from simple crossing counts \citep{karatzas2014brownian,morters2010brownian}.
\end{example}

\begin{example}[Irregular EHR labs and wearable streams]
For irregular observations \(Y_{ij}\) at times \(t_{ij}\), a linear Gaussian diffusion model gives
\[
X_i(t_{ij})\mid X_i(t_{i,j-1})
\sim
\mathcal{N}\{X_i(t_{i,j-1})+\mu_i\Delta_{ij},\sigma_i^2\Delta_{ij}\},
\qquad
Y_{ij}=X_i(t_{ij})+\epsilon_{ij}.
\]
The observation schedule changes the transition variance through \(\Delta_{ij}=t_{ij}-t_{i,j-1}\). This is why Brownian motion is a natural engine for EHR and wearable smoothing: uncertainty grows during long gaps and shrinks when data arrive densely \citep{kalman1960new,durbin2012time}.
\end{example}

\begin{example}[Cumulative score processes for model diagnostics]
Let \(U_i(\widehat\beta)\) be a fitted-model score residual ordered by event time, predicted risk, dose, or calendar time. A cumulative score process has the form
\[
S_n(t)=\frac{1}{\sqrt n}\sum_{i\le nt}U_i(\widehat\beta).
\]
Under a stable model, \(S_n(t)\) often converges to a Brownian motion or Brownian bridge after projection for estimated parameters. Persistent drift is a diagnostic: proportional hazards may be failing, an exposure effect may vary over time, or a subgroup may be systematically misfit \citep{andersen1993statistical,vandervaart1996weak,kosorok2008introduction}.
\end{example}

\begin{example}[Safety surveillance as a Brownian alarm]
Suppose \(O(t)\) is an observed cumulative adverse-event count and \(E(t)\) is its predictable expected count under standard care. A variance-scaled process
\[
M(t)=\frac{O(t)-E(t)}{\sqrt{\widehat{\operatorname{Var}}\{O(t)-E(t)\}}}
\]
is often approximated by a martingale with Brownian limits under no safety signal. A boundary crossing can then be interpreted as an alarm with controlled false-signal behavior. The same first-passage intuition from the reflection principle becomes a monitoring tool for pharmacovigilance, hospital quality, or vaccine safety \citep{fleming1991counting,jennison2000group}.
\end{example}

\section{Conclusion}

This manuscript has followed Brownian motion from its construction as a
probability measure on continuous paths to its use as a modelling and
diagnostic language in biostatistics. The same process supports several
apparently different tasks: latent biomarker diffusion, Functional PCA for
trajectory compression, degradation modelling through first-passage times,
dynamic frailty, group-sequential monitoring, calibration assessment,
survival and recurrent-event residuals, and state-space smoothing for EHR
and wearable streams.

The main lesson is not that biomedical systems are literally Brownian. The
more useful point is that Brownian motion supplies a disciplined stochastic
baseline. Its covariance kernel explains why the Karhunen-Loève expansion
and FPCA are natural tools for functional biomarkers. Its reflection
principle and hitting-time calculations clarify threshold endpoints. Its
local time separates near-cutoff burden from simple crossing counts. Its
martingale characterization links counting-process residuals to survival
diagnostics. Its SDE form, illustrated by the Black-Merton-Scholes model,
shows how continuous-time noise enters applied likelihoods and prediction
models.

For the working biostatistician, the Brownian bridge is especially useful
because it turns ordered residuals into a visual and inferential check on
model stability. A calibrated bridge path can reveal drift by risk score,
event time, clinic site, calendar period, sensor batch, or visit sequence.
The tool is modest, but that modesty is its strength: it asks the analyst
to declare the ordering, centering rule, variance scale, and dependence
structure before turning an apparent pattern into a scientific claim.

\section*{Data Availability}

No patient-level or protected health data are analyzed in this tutorial
manuscript. The figures and tables are methodological illustrations
generated from analytic formulas or simulation scripts included in the
manuscript workspace. The illustrative archival experiment uses Project
Gutenberg ebook 84 for \emph{Frankenstein}; the chapter-level feature table,
Monte Carlo summary, plotting scripts, and generated figures are included in
the workspace. Any empirical extension should provide the exact feature
definitions, preprocessing scripts, model code, simulation settings, and
diagnostic procedures in a public repository when privacy, copyright, and
data-use agreements permit.

\section*{AI Disclosure Statement}

This manuscript was revised with the assistance of an AI coding and writing assistant for language editing, structural revision, and drafting of mathematical exposition. The author is responsible for checking all mathematical claims, references, and interpretations before submission.

\bibliographystyle{apalike}
\bibliography{references}

\end{document}